\documentclass[letterpaper, 10 pt, conference]{ieeeconf}

\usepackage{cite}
\usepackage{graphicx}
\usepackage{amsfonts}
\usepackage{amsmath} % Needs to be defined before amsthm package
\usepackage{amssymb}

\usepackage{amsthm}	% Needs to be defined after amsmath package
\usepackage{dsfont}
\usepackage{enumerate}
\usepackage{mathrsfs}
\usepackage{cleveref}
\usepackage[labelformat=simple]{subcaption}

\usepackage{siunitx}
\usepackage[hyphens]{url}

\graphicspath{{figures/}}

\usepackage{geometry}
\newtheorem{thm}{Theorem}
\newtheorem{lem}[thm]{Lemma}
\newtheorem{prop}[thm]{Proposition}

\newtheorem{rem}{Remark}

\renewcommand{\arraystretch}{0.9}

\newcommand{\Bset}{\mathbb{B}}

\newcommand{\Rset}{\mathbb{R}}

\newcommand{\Uset}{\mathbb{U}}

\newcommand{\Xset}{\mathbb{X}}
\newcommand{\Yset}{\mathbb{Y}}
\newcommand{\Zset}{\mathbb{Z}}

\newcommand{\Ccal}{{\mathcal{C}}}

\newcommand{\Kcal}{{\mathcal{K}}}

\newcommand{\Ncal}{{\mathcal{N}}}

\newcommand{\Scal}{{\mathcal{S}}}
\newcommand{\Tcal}{{\mathcal{T}}}

\newcounter{l1}
\newcounter{l2}
\newcounter{l3}
\newcommand{\bdotlist}{\begin{list}{$\bullet$}{}}
\newcommand{\bboxlist}{\begin{list}{$\Box$}{}}
\newcommand{\bbboxlist}{\begin{list}{\raisebox{.005in}{{\tiny
$\blacksquare$ \ \ }}}{}}
\newcommand{\bdashlist}{\begin{list}{$-$}{} }
\newcommand{\blist}{\begin{list}{}{} }
\newcommand{\barablist}{\begin{list}{\arabic{l1}}{\usecounter{l1}}}
\newcommand{\balphlist}{\begin{list}{(\alph{l2})}{\usecounter{l2}}}
\newcommand{\bAlphlist}{\begin{list}{\Alph{l2}.}{\usecounter{l2}}}
\newcommand{\bdiamlist}{\begin{list}{$\diamond$}{}}
\newcommand{\bromalist}{\begin{list}{(\roman{l3})}{\usecounter{l3}}}

\newcommand{\indi}[1]{\mathds{1}  \hspace{-.03in} \left\{#1 \right\}}
\newcommand{\haus}[3][]{d_{\rm H}#1( #2, \, #3 #1)}

\IEEEoverridecommandlockouts
\begin{document}

% *** TITLE ***
\title{\LARGE \bf A Multi-Battery Model for the Aggregate Flexibility of Heterogeneous Electric Vehicles}

% *** AUTHORS and THANKS ***
\author{Feras Al Taha, \, Tyrone Vincent, \,   Eilyan Bitar % <-this % stops a space
\thanks{To appear in the Proceedings of the 2023 American Control Conference (ACC).}
\thanks{This work was supported in part by the Natural Sciences and Engineering Research Council of Canada, in part by the Cornell Atkinson Center for Sustainability, in part by The Nature Conservancy, and in part by the Office of Naval Research via grant N00014-20-1-2492.}% <-this % stops a space
\thanks{F. Al Taha is with the School of Electrical and Computer Engineering, Cornell University, Ithaca, NY, 14853, USA (e-mail:  foa6@cornell.edu). }%
\thanks{T. Vincent is with the Department of Electrical Engineering and Computer Science, Colorado School of Mines, Golden, CO 80401, USA (e-mail: tvincent@mines.edu).}%
\thanks{E. Bitar is with the School of Electrical and Computer Engineering, Cornell University, Ithaca, NY, 14853, USA (e-mail: eyb5@cornell.edu). }%
}

\maketitle 

% *** ABSTRACT ***
\begin{abstract}
The increasing prevalence of  electric vehicles (EVs) in the transportation sector   will introduce a large number of highly flexible electric loads that EV aggregators can pool and control to provide energy and ancillary services to the wholesale electricity market.
To integrate large populations of EVs into  electricity market operations, aggregators must express the aggregate flexibility of the EVs under their control in the form of a small number of energy storage (battery) resources that accurately capture the supply/demand capabilities of the individual EVs as a collective. 
To this end, we propose a novel multi-battery flexibility model defined as a linear combination of  a small number of base sets (termed \textit{batteries}) that reflect the differing geometric shapes of the individual EV flexibility sets, and suggest a clustering approach to identify these base sets. 
We study the problem of computing a multi-battery flexibility set that has minimum Hausdorff distance to the aggregate flexibility set, subject to the constraint that the multi-battery flexibility set be a subset of the aggregate flexibility set.  We show how to conservatively approximate this problem with a tractable convex program, and illustrate the performance achievable by our method with several numerical experiments.
\end{abstract}

% *** MAIN FILES ****
\section{Introduction} \label{sec:introduction}

The  growing adoption of electric vehicles (EVs) will sharply increase peak electricity demand, stressing both the electric power distribution and transmission systems if left unmanaged \cite{quiros2018electric,clement2009impact,muratori2018impact}. However, a growing number of field studies have shown  EV charging requirements in residential and workplace environments to be highly flexible in the sense that most EVs remain connected to their chargers far longer than the amount of time required to satisfy their energy needs \cite{alexeenko2021achieving,lee2019acn,smart2015plugged}. Given access to bidirectional charging infrastructure, the charging flexibility of many EVs can be centrally aggregated and dispatched to provide energy and ancillary services to the  wholesale electricity market, which has been made possible by FERC Order No. 2222 \cite{ferc2222}.

Since it would be impractical and computationally intractable for an independent system operator to individually dispatch each EV load seeking to participate in the wholesale electricity market, the flexibility of individual EVs  must be aggregated. In particular, under current market rules, an EV load aggregator participating in the wholesale electricity market must express the aggregate flexibility of the EVs under its command  as a single 
equivalent energy storage resource. These equivalent representations are typically defined in terms of time-varying upper and lower bounds on the aggregate charging/discharging power and aggregate state-of-charge of the collective \cite{nyiso,caiso}.
Characterizing the aggregate flexibility set associated with a collection of EVs requires the calculation of the Minkowski sum of the individual EVs' flexibility sets, which are typically expressed as convex polytopes in half-space representation. As the exact calculation of such Minkowski sums is  known to be computationally intractable in general \cite{gritzmann1993minkowski,trangbaek2012exact}, many  existing methods in the literature provide inner approximations of the aggregate flexibility set \cite{nayyar2013aggregate,hao2014aggregate,muller2015aggregation,zhao2016extracting,zhao2017geometric,nazir2018inner,kundu2018approximating}. Inner approximations are desirable as they are guaranteed to only contain feasible points.

\subsection{Related Literature and Contribution}
There is a  family of related methods in the literature that build inner approximations to the individual flexibility sets, utilizing specific geometries for the individual inner approximations which enable their efficient summation  to yield an inner approximation to the aggregate flexibility set.  
The inner approximations to the individual flexibility sets utilized by these methods are expressed as transformations of a common polyhedral set  (termed the \emph{base set}). 
For example, M{\"u}ller et al. \cite{muller2015aggregation} use a specific class of zonotopes (a family of centrally symmetric polytopes) to internally approximate each individual flexibility set. 
Zhao et al. \cite{zhao2017geometric} utilize homothetic transformations (dilation and translation) of a user-defined convex polytope. 
This approach was later extended by Al Taha et al. \cite{altaha2022efficient} to construct inner approximations via  general affine transformations of a given convex polytope. 
Nazir et al. \cite{nazir2018inner} provide a method to internally approximate the individual flexibility sets using unions of homothetic transformations of axis-aligned hyperrectangles.  
While all of the approximating sets provided by these methods have the property that their Minkowski sum can be easily calculated or internally approximated, their reliance on transformations of a \emph{single}  base set may  result in overly conservative approximations of the aggregate flexibility set when there is significant asymmetry and/or heterogeneity in the shapes of the individual flexibility sets. 

To limit this conservatism, we propose a novel  \emph{multi-battery flexibility model} that consists of a weighted sum of a small number of base sets that can better capture the differing geometric shapes of the individual flexibility sets, allowing a trade-off between model complexity and fidelity. 
Using this framework,  we study the problem of computing an optimal multi-battery flexibility set that has minimum Hausdorff distance to the aggregate flexibility set, subject to the constraint that the multi-battery flexibility set be contained within the aggregate flexibility set.
While this problem is shown to be computationally intractable, we provide a conservative approximation in the form of a convex program whose size scales polynomially with the number and dimension of the individual flexibility sets and base sets.
An advantage of using this model is that aggregate power profiles in the multi-battery set can be efficiently decomposed into individually feasible charging profiles for each EV using an affine mapping obtained as a byproduct of solving the proposed convex program.
We also suggest a clustering approach to identifying the base sets used in the multi-battery approximation to ensure that the different geometries of the individual sets are captured.
We conduct experiments to illustrate the improvement in performance that can be achieved in going from a single base set to multiple base sets using the multi-battery approximation method proposed in this paper. 

\subsection{Notation} 
We employ the following notational conventions throughout the paper.
Let $\Rset$ denote the set of real numbers. 
We denote the indicator function of set $\Scal$ by $\indi{x\in\Scal} = 1$ if $x \in \Scal$ and $\indi{x\in\Scal} = 0$ if $x \notin \Scal$.  
We denote the $n \times n$ identity matrix by $I_n$. 
Given a pair of matrices $A$ and $B$ of appropriate dimension, we let  ${\rm diag}(A, \, B)$  and $(A, \, B)$ denote the matrices formed by stacking $A$ and $B$ block-diagonally and vertically, respectively. Unless stated~otherwise, we let $\|\cdot\|$ denote an arbitrary norm and $\|\cdot\|_{*}$ its corresponding dual norm.
The Hausdorff distance between two sets $\Xset,\Yset \subseteq \Rset^n$ is defined  as  $\haus{\Xset}{\Yset} := \max \left \{ \sup_{x\in\Xset} \inf_{y \in \Yset} \|x-y\|, \sup_{y\in\Yset} \inf_{x\in\Xset} \|y-x\| \right \}$.
We refer to convex polytopes in half-space representation as \textit{H-polytopes} and to affine transformations of H-polytopes as \textit{AH-polytopes}.

\subsection{Paper Organization}
The remainder of the paper is organized as follows. 
In Section \ref{sec:model}, we introduce the multi-battery flexibility model and state the problem addressed in this paper. 
In Section~\ref{sec:cluster}, we provide a clustering approach to identifying the base sets employed in the multi-battery flexibility model. In Section~\ref{sec:approx}, we develop a scalable convex optimization model to compute multi-battery approximations. 
Numerical experiments are provided in Section \ref{sec:experiments} and 
Section \ref{sec:conclusion} concludes the paper.

\section{Problem Formulation}
\label{sec:model}
In this section, we present the EV flexibility set model and formulate the problem of computing optimal inner approximations of an aggregate flexibility set using multi-battery flexibility sets. 

\subsection{EV Charging Dynamics}

Consider  a population of $N$  EVs indexed by $i\in\Ncal:=\{1,\dots,N\}$.
Time is discretized into periods of equal length $\delta>0$ and indexed by $t\in\Tcal:=\{0,\dots,T-1\}$.
We let $u_i(t)$ denote the charging rate (kW) of EV  $i \in \Ncal$ at time $t \in \Tcal$, and $u_i := (u_i(0), \, \dots, \, u_i(T-1)) \in \Rset^T$ its charging profile. Given a charging profile $u_i$, the cumulative energy (kWh) supplied to each  EV  $i \in \Ncal$ is assumed to evolve according to the difference equation
\begin{align}
x_i(t+1) = x_i(t) + u_i(t) \delta , \quad t\in \Tcal,
\end{align}
where $x_i(0) = 0$.   We denote the resulting cumulative energy profile by $x_i = (x_i(1), \dots, x_i(T)) \in \Rset^{T}$, which satisfies the relationship 
\begin{align*}
 x_i = Lu_i,
\end{align*} 
where $ L \in \Rset^{T \times T}$ is a lower triangular matrix given by $L_{ij} := \delta$ for all $j \leq i$. 

\subsection{Individual and Aggregate Flexibility Sets}

The \textit{individual flexibility set} associated with each EV $i \in \Ncal$ is defined as the set of all admissible charging profiles, denoted by
\begin{align} \label{eq:indi_set}
\Uset_i : = \left\{ u \in \Rset^T \, | \, Hu \le h_i  \right\},
\end{align}
where $H := (L, -L, \, I_T, -I_T)$ and $h_i := (\overline{x}_i, -\underline{x}_i,  \overline{u}_i, -\underline{u}_i )$.  
The vectors  $\underline{u}_i , \overline{u}_i \in \Rset^T$ represent  minimum and maximum power limits on the charging profile, respectively. 
The vectors  $\underline{x}_i , \overline{x}_i \in \Rset^T$ represent  minimum and maximum energy limits on the cumulative energy profile, respectively. 
We assume that every individual flexibility set is nonempty.
Flexibility sets of the form \eqref{eq:indi_set} are compact, convex polytopes which are commonly referred to as \textit{generalized} or \textit{virtual battery models} in the literature \cite{hao2014aggregate,zhao2017geometric}. 

\begin{rem}[Alternative representation] \rm \label{rem:sparse}
An individual flexibility set can be equivalently represented in terms of the corresponding set of cumulative energy profiles, given by
\begin{align*}
    \Xset_i := L \Uset_i = \{ x \in\Rset^T \,|\, HL^{-1} x \le h_i \}.
\end{align*}
These alternative representations may be advantageous from a computational perspective, as the matrix $H L^{-1}$ is much sparser than the matrix $H$.
\end{rem}

We define the \textit{aggregate flexibility set} of a population of EVs as the Minkowski sum of the individual flexibility sets, denoted by
\begin{align} \label{eq:agg_set}
\Uset := \sum_{i \in \Ncal} \Uset_i. 
\end{align}
Computing the Minkowski sum of multiple H-polytopes is NP-hard in general \cite{tiwary2008hardness}. 
We note that it is straightforward to construct an outer approximation of the aggregate flexibility set by  simply adding the  right-hand side vectors of the individual flexibility sets, i.e.,  it can be shown that $    \Uset \subseteq \{ u \in \Rset^T \, | \, Hu \leq \sum_{i\in\Ncal} h_i  \}$. Since outer approximations may contain infeasible aggregate charging profiles, this motivates, in part, the need for efficient methods to construct accurate inner approximations.

\subsection{Multi-Battery Flexibility Set}

Acknowledging the difficulty in characterizing the aggregate flexibility set via exact Minkowski sum computation, our primary objective in this paper is to develop an efficient convex optimization based method to compute accurate inner approximations of the aggregate flexibility set.

Recognizing that the individual flexibility sets can differ greatly in terms of  shape, we seek to approximate their Minkowski sum by the weighted aggregation of a much smaller number of $ K \ll N$ virtual batteries (as defined in \eqref{eq:indi_set}) that effectively capture these different shapes.
We refer to these representative sets as the \textit{base sets} and denote them by 
\begin{align*}
\Bset_k := \{  u \in \Rset^T \, | \, Hu \le b_k\}, \quad k\in\Kcal,
\end{align*}
where $\Kcal:=\{1,\dots,K\}$ and $b_k \in\Rset^{4T}$ denotes the right-hand side  data of each base set.
In Section~\ref{sec:cluster},  we suggest
a clustering approach to constructing  the base sets  in a manner that seeks to capture the different geometries of the individual flexibility sets.

We refer to a  weighted combination of these sets as a \textit{multi-battery flexibility set}, defined as
\begin{align} \label{eq:MBM}
    \Bset := \mu + \sum_{k\in\Kcal} \sigma_k \Bset_k,
\end{align}
where $\mu\in\Rset^T$ denotes a translation vector, and $\sigma_k\in\Rset_+$ denotes the scaling factor applied to the $k$-th base set for $k=1,\dots,K$.
We note that homothetic transformations of the individual base sets $\Bset_k$ preserve their battery structure since the transformed base sets can be written as
\begin{align} \label{eq:homothet}
    \beta + \alpha \Bset_k = \left\{ u \in \Rset^T \, | \, Hu \le H \beta  + \alpha b_k  \right\},
\end{align}
for any translation  $\beta\in\Rset^T$ and  scaling $\alpha\in\Rset_+$.
Consequently, a multi-battery flexibility set can be interpreted as the aggregation of $K$ virtual batteries.

\subsection{Approximating the Aggregate Flexibility Set}
Given a collection of individual flexibility sets $\Uset_1,\dots,\Uset_N$ and base sets $\Bset_1,\dots, \Bset_K 
$, we seek to compute a multi-battery flexibility set $\Bset$ that has
minimum Hausdorff distance to the aggregate flexibility set,
subject to the requirement that it 
be contained within the aggregate flexibility set. This can be expressed according to the following polytope containment problem:
    \begin{align} \label{eq:problem}
     \text{minimize}  \ \, \haus{\Bset}{\Uset}  \ \, \text{subject to} \ \ \Bset = \mu + \sum_{k\in\Kcal} \sigma_k \Bset_k  \subseteq \Uset, 
    \end{align}
where  the optimization variables are  $\mu \in \Rset^T$ and $\sigma_k \in \Rset_+$ for $k=1,\dots,K$. 
We refer to this problem as the \textit{multi-battery approximation problem}. In Section \ref{sec:approx}, we discuss the hardness of this optimization problem, and suggest an approach to  conservatively approximate problem $\eqref{eq:problem}$ by a convex program.
In Section \ref{sec:cluster}, we first propose a clustering method to compute the  base sets which parameterize the multi-battery flexibility set.
 
\section{Clustering Flexibility Sets} \label{sec:cluster}

In this paper, we adopt a clustering-based approach to constructing the base sets  $\Bset_1, \dots, \Bset_K$ utilized in the multi-battery model \eqref{eq:MBM}. We partition the individual flexibility sets into $K$ clusters $\Ccal_1,\dots,\Ccal_K \subset \Ncal$, where flexibility sets belonging to the same cluster are more similar to each other than to flexibility sets belonging to other clusters. 
More formally, using the squared Hausdorff distance as a measure of similarity between flexibility sets, we seek to solve the following clustering problem 
\begin{align} \label{eq:poly_cluster}
    \textrm{minimize} \ \sum_{k\in\Kcal} \sum_{i \in \Ccal_k} \haus{\Uset_i}{\mathbb{B}_k}^2
\end{align}
with respect to the optimization variables $\Ccal_k$ and $\Bset_k$ for $k=1,\dots,K$. 

While natural to state, the polytope clustering problem~\eqref{eq:poly_cluster} is computationally intractable to solve, in part because it is NP-hard to compute the Hausdorff distance between two H-polytopes \cite{konig2014computational}. 
To address this issue, we utilize the Lipschitz continuity of H-polytopes (with respect to perturbations in the right-hand side data) to construct a tractable alternative for the clustering problem \eqref{eq:poly_cluster}, which can be solved using standard clustering algorithms from the literature.  The following result taken from \cite[Theorem 2.4]{li1993sharp} provides a sharp characterization of Lipschitz constants for H-polytopes.

\begin{lem} \label{lem:continuity} \rm 
Let $\Xset = \{ x \in \Rset^{n} \, | \, A x \leq b_x\}$ and $\Yset = \{ y \in \Rset^{n} \, | \, A y \leq b_y\}$  be nonempty and compact H-polytopes. If $\|\cdot\|_p$, $\|\cdot\|_q$ are two arbitrary norms and $d_H(\cdot, \cdot)$ is the Hausdorff metric induced by the norm $\|\cdot\|_p$, then 
\begin{align*}
    \haus{\Xset}{\Yset} \le L(A) \| b_x - b_y \|_q,
\end{align*}
where 
\begin{align}
   L(A) := \sup \left\{ \|x\|_{q*}  \left| \ \text{\begin{minipage}{13em} $\|A^\top x\|_{p*} = 1$; $x\ge 0$; the rows of $A$ corresponding to nonzero components of $x$ are linearly independent.  \end{minipage}} \right. \right\}.  
   \end{align}
\end{lem}

Using Lemma \ref{lem:continuity}, we can upper bound the  clustering loss function in \eqref{eq:poly_cluster} by 
\begin{align*}
\sum_{k\in\Kcal} \sum_{i \in \Ccal_k} \haus{\Uset_i}{\Bset_k}^2 \leq L(H)^2 \sum_{k\in\Kcal} \sum_{i \in \Ccal_k} \|h_i - b_k\|_2^2,
\end{align*}
where the upper bound is specified in terms of the Euclidean norm.              
 Utilizing this upper bound as a surrogate for the original clustering loss function, we arrive at the following $K$-means clustering problem in the right-hand side data of the individual flexibility sets:
 partition the right-hand side data $\{h_i\}_{i=1}^N$ into $K$ clusters $\Ccal_1, \dots, \Ccal_K$
to minimize the within-cluster variation
\begin{align} \label{eq:rhs_clustering}
    \textrm{minimize} \ \sum_{k\in\Kcal} \sum_{i \in \Ccal_k} \| h_i - b_k \|_2^2,
\end{align}
where $b_k$ is the mean of the elements belonging to the $k$-th cluster, which is given by 
\begin{align} \label{eq:cluster_RHS}
    b_k := \frac{1}{|\Ccal_k|} \sum_{i \in \Ccal_k} h_i
\end{align}
for $k=1, \dots, K$.
Although the $K$-means clustering problem  is known to be NP-hard \cite{dasgupta2008hardness}, there are a number of algorithms, such as Lloyd's algorithm \cite{lloyd1982least}, that can be utilized to efficiently compute suboptimal solutions.

We  note that, given any clustering of the individual flexibility sets, the  base sets $\{\Bset_k\}_{k \in \Kcal}$ defined by right-hand side vectors given by the corresponding centroids in \eqref{eq:cluster_RHS} are guaranteed to be nonempty. This follows from the fact (which is straightforward to verify) that 
\begin{align} \label{eq:outer_approx}
    \frac{1}{|\Ccal_k|} \sum_{i \in \Ccal_k} \Uset_i \subseteq \Bset_k 
\end{align}
for all $k \in \Kcal$. In other words, for each cluster $\Ccal_k$, the corresponding base set $\Bset_k$ is guaranteed to be an outer approximation of the Minkowski average of the individual flexibility sets belonging to that cluster. Thus,  nonemptiness of the individual flexibility sets  guarantees nonemptiness of the  base sets for any clustering. 

We also note that for $K=1$, the base set obtained by this method corresponds to the base set originally proposed in \cite{zhao2017geometric} and later considered by \cite{altaha2022efficient} for single-battery approximations of aggregate flexibility sets.

\section{Multi-battery Approximation} \label{sec:approx}

The multi-battery approximation problem \eqref{eq:problem} is computationally intractable to solve for two fundamental reasons.
First, calculating the Hausdorff distance between Minkowski sums of H-polytopes is NP-hard in general \cite{konig2014computational}. Second, verifying  the containment of a Minkowski sum of H-polytopes in  the Minkowski sum of another set of H-polytopes is an NP-complete problem \cite{tiwary2008hardness}.
In Sections \ref{sec:approx_cont} and \ref{sec:approx_haus}, we address these computational issues by conservatively approximating both the containment condition and objective function  to yield a convex programming inner approximation to the multi-battery approximation problem \eqref{eq:problem}. 
In Section \ref{sec:disagg}, we show how to efficiently disaggregate any element in the multi-battery flexibility set into a collection of elements belonging to the individual flexibility sets.

\subsection{Approximating the Containment Condition: $\Bset\subseteq\Uset$} \label{sec:approx_cont}

We begin by providing a sufficient condition for the containment of a multi-battery set $\Bset$  within the aggregate flexibility set $\Uset$. In short, Proposition \ref{prop:ind_cont} 
replaces the original set containment constraint $ \mu +  \sum_{k\in\Kcal}  \sigma_k \Bset_k \subseteq \sum_{i \in \Ncal} \Uset_i$ with a more conservative collection of individual flexibility set containment constraints given by $\gamma_i + \sum_{k\in\Kcal} \Gamma_{k,i} \Bset_k \subseteq \Uset_i$  for $i=1, \dots, N$. 
The key simplifying element in Proposition \ref{prop:ind_cont} is that each of these individual set containment conditions only involves a single H-polytope on the right-hand side, rather than a Minkowski sum of H-polytopes.

\begin{prop} \label{prop:ind_cont} \rm  It holds that $\mu + \sum_{k\in\Kcal} \sigma_k \Bset_k \subseteq \Uset$ if there exist $\gamma_i\in\Rset^T$ and $\Gamma_{k,i}\in\Rset^{T\times T}$ for $i=1,\dots,N$ and $k=1,\dots,K$ such that
\begin{align}
    \label{eq:B_inclusion_1}&\mu = \sum_{i\in\Ncal} \gamma_i,\\
    \label{eq:B_inclusion_2}&\sigma_k I_T = \sum_{i\in\Ncal} \Gamma_{k,i}, \quad  \forall \, k \in \Kcal,\\
    \label{eq:B_inclusion_3}&\gamma_i + \sum_{k\in\Kcal} \Gamma_{k,i} \Bset_k \subseteq \Uset_i, \quad \forall \, i \in \Ncal.
\end{align}
\end{prop}

\begin{proof}
The desired result follows from the following string of inclusions:
\begin{align*}
\mu + \sum_{k\in\Kcal} \sigma_k \Bset_k &\overset{(a)}{=} \Big( \sum_{i\in\Ncal} \gamma_i  \Big) + \sum_{k\in\Kcal} \Big( \sum_{i\in\Ncal} \Gamma_{k,i} \Big)\Bset_k\\
&\overset{(b)}{\subseteq} \sum_{i\in\Ncal} \Big ( \gamma_i + \sum_{k\in\Kcal} \Gamma_{k,i} \Bset_k  \Big) \\
&\overset{(c)}{\subseteq} \sum_{i\in\Ncal} \Uset_i\\
& = \Uset.
\end{align*}
Equality $(a)$ follows from a direct substitution according to the conditions \eqref{eq:B_inclusion_1}-\eqref{eq:B_inclusion_2}. Inclusion $(b)$ follows from the fact that any element  belonging to the set $ ( \sum_{i\in\Ncal} \gamma_i  ) + \sum_{k\in\Kcal} ( \sum_{i\in\Ncal} \Gamma_{k,i} ) \Bset_k$ can be expressed as the sum of elements belonging to the sets $\gamma_i + \sum_{k \in \Kcal} \Gamma_{k,i} \Bset_k$ for $i = 1, \dots, N$.  Inclusion $(c)$ follows from the containment conditions  in \eqref{eq:B_inclusion_3}.
\end{proof}

We note that, unlike the approximating set $\Bset$, the internal approximations  to the individual flexibility sets  in Proposition \ref{prop:ind_cont}  are not required to be multi-batteries. 
They can be expressed as general affine transformations of the given base sets (cf. condition \eqref{eq:B_inclusion_3}), requiring only that their aggregation be a multi-battery flexibility set (cf. conditions \eqref{eq:B_inclusion_1}-\eqref{eq:B_inclusion_2}). This allows for  a broader family of approximations to the individual flexibility sets.

The set containment conditions \eqref{eq:B_inclusion_3} specified in Proposition \ref{prop:ind_cont}  can be  equivalently reformulated as a set of linear feasibility conditions using a well-known result from the literature \cite{sadraddini2019linear,kellner2015containment,rakovic2007optimized,altaha2022efficient}, which provides necessary and sufficient conditions for the containment of an AH-polytope within an H-polytope. 
\begin{lem} \label{lem:ah-polytope containment} \rm Let $\Xset = \{ x \in \Rset^{n_x} \, | \, H_x x \leq h_x\}$ and $\Yset = \{ y \in \Rset^{n_y} \, | \, H_y y \leq h_y\}$,   where $H_x \in \Rset^{m_x \times n_x}$, $H_y \in \Rset^{m_y \times n_y}$, and $\Xset$ is assumed to be nonempty. Given a vector $\gamma \in \Rset^{n_y} $ and matrix $\Gamma \in \Rset^{n_y \times n_x}$, it holds that $\gamma + \Gamma \Xset \subseteq \Yset$  if and only if there exists a matrix $\Lambda \in \Rset^{m_y \times m_x}$ such that
\begin{align}
\label{eq:inner cond 1} & \Lambda  \geq 0,  \\
\label{eq:inner cond 2}  &\Lambda  H_x = H_y\Gamma ,    \\
 \label{eq:inner cond 3}   &\Lambda h_x  \leq h_y -  H_y \gamma.
\end{align}
\end{lem}
Lemma \ref{lem:ah-polytope containment} follows from standard duality results in convex analysis,
and can be interpreted as a variant of Farkas’ Lemma. We refer the reader to \cite[Lemma 1]{altaha2022efficient} for a concise proof of Lemma \ref{lem:ah-polytope containment}. 

Utilizing Proposition \ref{prop:ind_cont} and Lemma \ref{lem:ah-polytope containment}, we  provide a set of linear feasibility conditions that are sufficient for the multi-battery containment condition $\Bset = \mu + \sum_{k \in \Kcal} \sigma_k \Bset_k \subseteq \Uset$.

\begin{thm} \rm \label{thm:Bi_in_Ui} 
It holds that $\mu + \sum_{k\in\Kcal} \sigma_k \Bset_k \subseteq \Uset$ if there exist $\Lambda_{k,i} \in \Rset^{4T\times4T}$, $\gamma_i\in\Rset^T$ and $\Gamma_{k,i}\in\Rset^{T\times T}$ for $i=1,\dots,N$ and $k=1,\dots,K$ such that
\begin{align}
\label{eq:sum_AH_in_H_a}&\mu = \sum_{i\in\Ncal} \gamma_i,\\
\label{eq:sum_AH_in_H_b}&\sigma_k I_T = \sum_{i\in\Ncal} \Gamma_{k,i}, \quad \forall \, k \in \Kcal,\\
\label{eq:sum_AH_in_H_1}& \Lambda_{k,i}  \geq 0, \quad \forall \, k \in \Kcal,   \, i \in \Ncal, \\
\label{eq:sum_AH_in_H_2}&\Lambda_{k,i}  H = H \Gamma_{k,i}, \quad   \forall \, k \in \Kcal, \, i \in \Ncal,\\
\label{eq:sum_AH_in_H_3}&\sum_{k\in\Kcal} \Lambda_{k,i} b_k  \leq h_i -  H \gamma_i, \quad \forall \, i \in \Ncal.
\end{align}
\end{thm}

\begin{proof}
From Proposition \ref{prop:ind_cont}, it suffices to show that $\gamma_i + \sum_{k\in\Kcal} \Gamma_{k,i} \Bset_k \subseteq \Uset_i$ for $i=1,\dots,N$ to prove the desired result.
For $i\in\Ncal$, the Minkowski sum under consideration can be rewritten as an AH-polytope given by
\begin{align*}
    \gamma_i + \sum_{k\in\Kcal} \Gamma_{k,i} \Bset_k = \gamma_i + \bar{\Gamma}_i \bar{\Bset},
\end{align*}
where $\bar{\Gamma}_i := \begin{bmatrix}\Gamma_{1,i} &\hdots &\Gamma_{K,i}\end{bmatrix}$ and  $\bar{\Bset} := \{\bar{u} \in \Rset^{KT} \, | \, {\rm diag}(H, \dots, H) \bar{u} \le (b_1,\dots,b_K) \}$. 
From Lemma \ref{lem:ah-polytope containment}, it follows  that conditions \eqref{eq:sum_AH_in_H_1}-\eqref{eq:sum_AH_in_H_3}  are necessary and sufficient for the  containment  $\gamma_i + \bar{\Gamma}_i \bar{\Bset} \subseteq \Uset_i$ for $i=1,\dots,N$. 
\end{proof}

\subsection{Approximating the Hausdorff Distance: $\haus{\Bset}{\Uset}$} \label{sec:approx_haus}

As noted earlier, computing the Hausdorff distance between the multi-battery set $\Bset$ and the aggregate flexibility set $\Uset$ exactly is an NP-hard problem.
To address this, we seek to construct an upper bound on this distance that can be utilized as a surrogate for the objective function of the multi-battery approximation problem \eqref{eq:problem}.
We derive an upper bound by utilizing the Lipschitz continuity of H-polytopes with respect to their right-hand side data, and exploiting the fact that each base set is a superset of the Minkowski average of the individual flexibility sets that belong to its corresponding cluster, i.e., $\Bset_k \supseteq (1/|\Ccal_k|) \sum_{i \in \Ccal_k} \Uset_i$.

\begin{thm} \rm \label{thm:obj_fnct} 
Given a collection of clusters $\{\Ccal_k\}_{k \in \Kcal}$, base sets $\{\Bset_k\}_{k \in \Kcal}$, and a multi-battery $\Bset = \mu + \sum_{k \in \Kcal} \sigma_k \Bset_k$, it holds that
\begin{align} \label{eq:upper_bound}
    \haus{\Bset}{\Uset} \le L(H) \sum_{k\in\Kcal} \| H \mu_k + ( \sigma_k - |\Ccal_k|  )   b_k \|,
\end{align} 
where $\mu_1,\dots,\mu_K\in\Rset^T$  is any collection of   vectors that satisfy $\sum_{k\in\Kcal}\mu_k = \mu$.
\end{thm}

\begin{proof}
The desired result follows from the following string of inequalities:
\begin{align*}
    \haus{\Bset}{\Uset} &= \haus[\Big]{\sum_{k\in\Kcal} \mu_k + \sigma_k \Bset_k}{\sum_{k\in\Kcal} \sum_{i\in \Ccal_k} \Uset_i}\\
    &\overset{(a)}{\le}\haus[\Big]{\sum_{k\in\Kcal} \mu_k+\sigma_k \Bset_k}{\sum_{k\in\Kcal} |\Ccal_k| \Bset_k}\\ 
    &\overset{(b)}{\le} \sum_{k\in\Kcal} \haus{\mu_k + \sigma_k \Bset_k}{|\Ccal_k| \Bset_k} \\ 
    &\overset{(c)}{\le} L(H) \sum_{k\in\Kcal}  \left \| H \mu_k +  \sigma_k b_k  - |\Ccal_k|    b_k \right \|.
\end{align*}
Inequality $(a)$ follows from $\Bset \subseteq \Uset \subseteq \sum_{k\in\Kcal} |\Ccal_k|\Bset_k$ since $\sum_{i\in\Ccal_k} \Uset_i \subseteq |\Ccal_k|\Bset_k$ by \eqref{eq:outer_approx} and from the fact that given three nonempty and compacts sets $\Xset,\Yset,\Zset\subseteq\Rset^T$ such that $\Xset\subseteq\Yset\subseteq\Zset$, we have $\haus{\Xset}{\Yset}\le\haus{\Xset}{\Zset}$. 
Inequality $(b)$ follows from the triangle inequality for the Hausdorff distance between nonempty, compact sets.
Inequality $(c)$ follows from Lemma \ref{lem:continuity} and from using an H-polytope representation of $\mu_k+\sigma_k\Bset_k$ as in \eqref{eq:homothet}. 
\end{proof}

Using the sufficient containment conditions \eqref{eq:sum_AH_in_H_a}-\eqref{eq:sum_AH_in_H_3} in Theorem \ref{thm:Bi_in_Ui} and the surrogate objective function \eqref{eq:upper_bound} in Theorem \ref{thm:obj_fnct}, we obtain the following conservative approximation of the multi-battery approximation problem~\eqref{eq:problem}:
\begin{align}
\nonumber \textrm{minimize}\ \, \quad &  \sum_{k\in\Kcal} \left \| H \mu_k + \left( \sigma_k - |\Ccal_k| \right )   b_k \right \|\\
\label{eq:convex_program} \textrm{subject to } \quad &  \sum_{k\in\Kcal} \mu_k = \sum_{i\in\Ncal} \gamma_i,\\
\nonumber &\sigma_k I_T = \sum_{i\in\Ncal} \Gamma_{k,i}, \quad \forall k\in\Kcal,\\
\nonumber &\Lambda_{k,i}  \geq 0, \quad \forall k\in\Kcal, i\in\Ncal, \\
\nonumber &\Lambda_{k,i}  H = H \Gamma_{k,i}, \quad \forall k\in\Kcal, i\in\Ncal, \\
\nonumber &\sum_{k\in\Kcal} \Lambda_{k,i} b_k  \leq h_i -  H \gamma_i, \quad \forall i\in\Ncal.
\end{align}

Problem \eqref{eq:convex_program} is a convex program with linear constraints in the decision variables $\gamma_i$, $\Gamma_{k,i}$, $\mu_k$, $\sigma_k$ and $\Lambda_{k,i}$ for $k=1,\dots,K$ and $i=1,\dots,N$. 
The number of decision variables and constraints in this problem scale  polynomially in $N$, $K$, and $T$.  

In the special case where each individual flexibility set is a homothet of the base set associated with its cluster (i.e., for each $k \in \Kcal$ and $i\in\Ccal_k$ there exist  $\gamma_i \in \Rset^T$ and $\beta_i\in\Rset_+$ such that  $\Uset_i = \gamma_i + \beta_i \Bset_k$), then the solution obtained from solving the convex program \eqref{eq:convex_program} will equal the aggregate flexibility set exactly.

\begin{rem}[Cluster-wise problem decomposition] \rm 
We note that one can reduce the number of decision variables and constraints in problem \eqref{eq:convex_program} by replacing the original containment condition $\mu + \sum_{k \in \Kcal} \sigma_k \Bset_k \subseteq \sum_{i \in \Ncal} \Uset_i$ with a more restrictive set of cluster-wise containment conditions 
$\mu_k + \sigma_k \Bset_k \subseteq \sum_{i\in\Ccal_k} \Uset_i$ for $k=1,\dots,K$. Under this additional restriction, problem \eqref{eq:convex_program}  decomposes into  $K$ separate inner approximation problems, one for each cluster. Naturally, this reduction in complexity may be accompanied by additional conservatism.
\end{rem}

\subsection{Disaggregating Aggregate Power Profiles} \label{sec:disagg}

To actually implement an aggregate power profile $u\in\Uset$ in the aggregate flexibility set, one must decompose $u$ into a collection of individually feasible power profiles $u_i\in\Uset_i$ for $i=1,\dots,N$ such that $u=\sum_{i\in\Ncal}u_i$.
Of course, this disaggregation can be carried by solving a linear feasibility problem whose size grows with the population size $N$.

One can bypass having to solve such linear feasibility problems by utilizing the class of multi-battery approximations proposed in this paper to disaggregate any power profile in a multi-battery flexibility set using an affine mapping that is computed as a byproduct of solving the convex program \eqref{eq:convex_program}. Consider a multi-battery flexibility set  $\Bset=\mu+\sum_{k\in\Kcal}\sigma_k \Bset_k$ satisfying conditions  \eqref{eq:sum_AH_in_H_a}-\eqref{eq:sum_AH_in_H_3} in Theorem~\ref{thm:Bi_in_Ui}.
Let $u=\mu+\sum_{k\in\Kcal}\sigma_k \tilde{u}_k \in\Bset$ be an arbitrary power profile in this set, where $\tilde{u}_k \in\Bset_k$ for $k=1,\dots,K$.
It follows from \eqref{eq:sum_AH_in_H_a}-\eqref{eq:sum_AH_in_H_b} that 
\begin{align} \label{eq:disagg}
    \mu + \sum_{k\in\Kcal}\sigma_k \tilde{u}_k = \sum_{i\in\Ncal} \Big( \gamma_i + \sum_{k\in\Kcal} \Gamma_{k,i} \tilde{u}_k \Big ).
\end{align}
Since conditions \eqref{eq:sum_AH_in_H_1}-\eqref{eq:sum_AH_in_H_3} imply that $\gamma_i+\sum_{k\in\Kcal}\Gamma_{k,i} \Bset_k \subseteq \Uset_i$ for all $i\in\Ncal$ (a consequence of Lemma \ref{lem:ah-polytope containment}), it must be that $\gamma_i + \sum_{k\in\Kcal} \Gamma_{k,i} \tilde{u}_k  \in \Uset_i$ for all $i\in\Ncal$.
Accordingly, the aggregate power profile $u\in\Bset$ can be disaggregated into a collection of individually realizable power profiles given by
\begin{align}
    u_i := \gamma_i + \sum_{k\in\Kcal} \Gamma_{k,i} \tilde{u}_k 
\end{align}
for $i=1,\dots,N$.
 
\section{Experiments} \label{sec:experiments}

In this section, we illustrate the method proposed in this paper 
using simulated EV charging requirements based on a daytime workplace charging scenario.
We first show how to express the charging requirements of an  EV as a flexibility set, as defined in Eq. \eqref{eq:indi_set}.
Using this model, we then conduct experiments to illustrate the improvement in performance that can be achieved in going from a single base set (such as in \cite{zhao2017geometric} and \cite{altaha2022efficient})  to multiple base sets in the multi-battery approximation.

\subsection{Simulating EV Charging Requirements}

We consider a setting where each EV $i\in\Ncal$ plugs in to charge when arriving at time $a_i\in\Tcal$ and remains connected until its departure time $d_i\in\Tcal$.
For all time periods $t\in\{a_i,\dots,d_i\}$, we assume that EV $i$ can be charged at any rate between zero and $R_i\in\Rset_+$.
At any other time, the charging rate is required to be zero.
Each EV $i$ is assumed to request a total energy amount $E_i\in\Rset_+$ that must be delivered by its charging completion deadline.
Together, these charging constraints and energy requirements can be encoded as a virtual battery model \eqref{eq:indi_set} by specifying the charging profile limits $(\overline{u}_i,\underline{u}_i)$ and the cumulative energy profile limits $(\overline{x}_i,\underline{x}_i)$ according to
\begin{align}
    \label{eq:u_upper} \overline{u}_i(t) & = R_i \cdot \indi{ t \in \{a_i, \dots , d_i\}}\\
    \label{eq:u_lower} \underline{u}_i(t) & = 0 
\end{align}
    for $t=0,\dots,T-1$, and
\begin{align}
    \label{eq:x_upper} \overline{x}_i(t)  & =  \min(E_i,\delta R_i(t-a_i)) \cdot\indi{ t \geq a_i}\\
     \nonumber \underline{x}_i(t) & =  \max(0,E_i-\delta R_i(d_i-t)) \cdot \indi{ t \in \{a_i, \dots , d_i\}}\\
    \label{eq:x_lower}& \quad + E_i \cdot \indi{t > d_i}
\end{align}
for $t = 1,  \dots, T$. 
The upper and lower power limits in \eqref{eq:u_upper}-\eqref{eq:u_lower} restrict the charging power to be between zero and the maximum charging rate when the EV is plugged in, and zero when disconnected.
The upper and lower energy limits in \eqref{eq:x_upper}-\eqref{eq:x_lower} ensure that the EV receives its desired amount of energy by its charging completion deadline. 

Table \ref{tab:data_gen} summarizes the EV charging parameters used to simulate different scenarios of daytime workplace EV charging used in the numerical experiments of the next section.

\begin{table}[ht]
    \centering
    \begingroup
    \setlength{\tabcolsep}{6pt} 
    \renewcommand{\arraystretch}{1.2} 
    \begin{tabular}{|c|l|c|}
        \hline
        \textbf{Param.} & \textbf{Description} & \textbf{Value or Range} \\ \hline
        $\delta$ & Time period length  &  2/3 hr \\
        $T$ & Number of time periods & 18\\
        $a_i$ & Plug-in time & $[7 \text{ AM}, \, 10 \text{ AM}]$\\
        $d_i$ & Charging deadline & $[4 \text{ PM}, \, 7 \text{ PM}]$\\
        $R_i$ & Maximum charging rate  & $[7,13]$ kW\\
        $E_i$ & Desired amount of energy  & $[0,\delta R_i (d_i-a_i)]$ kWh\\ \hline
    \end{tabular}
    \endgroup
    \caption{Summary of EV charging parameters used in numerical experiments. The parameters are either fixed at the specified value or  uniformly distributed random variables over the specified interval. }
    \label{tab:data_gen}
\end{table}

\subsection{Peak Load Minimization}

\begin{figure*}[ht!]
    \centering
    \subfloat[Plot of the aggregate EV load profile under unmanaged charging, optimal charging, and suboptimal charging using the multi-battery approximation $(K=5)$.]{\includegraphics[width=0.32\linewidth]{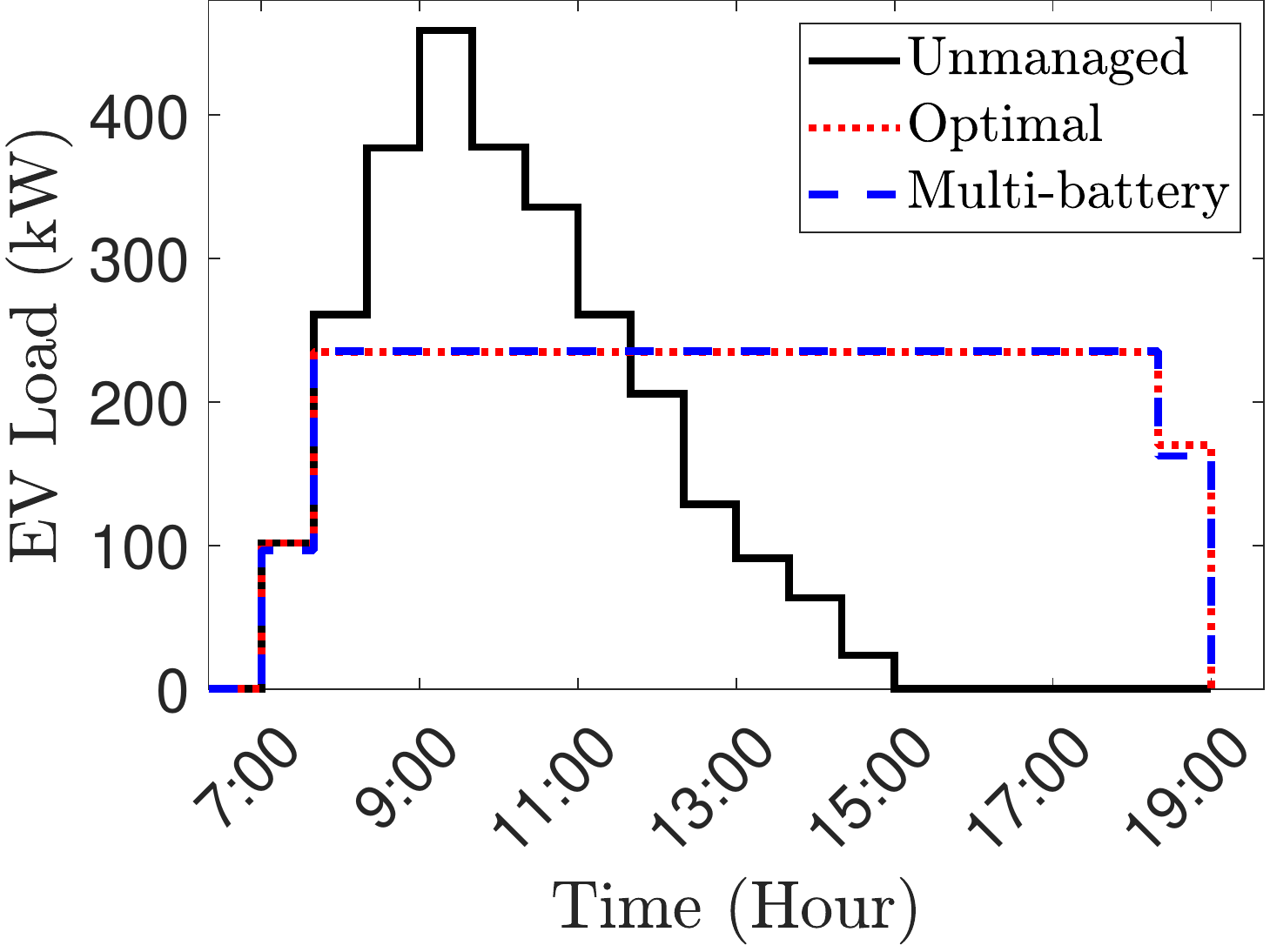}\label{fig:all_loads}}\quad
    \subfloat[Disaggregation of the EV load profile across the $K=5$ virtual batteries.]{\includegraphics[width=0.32\linewidth]{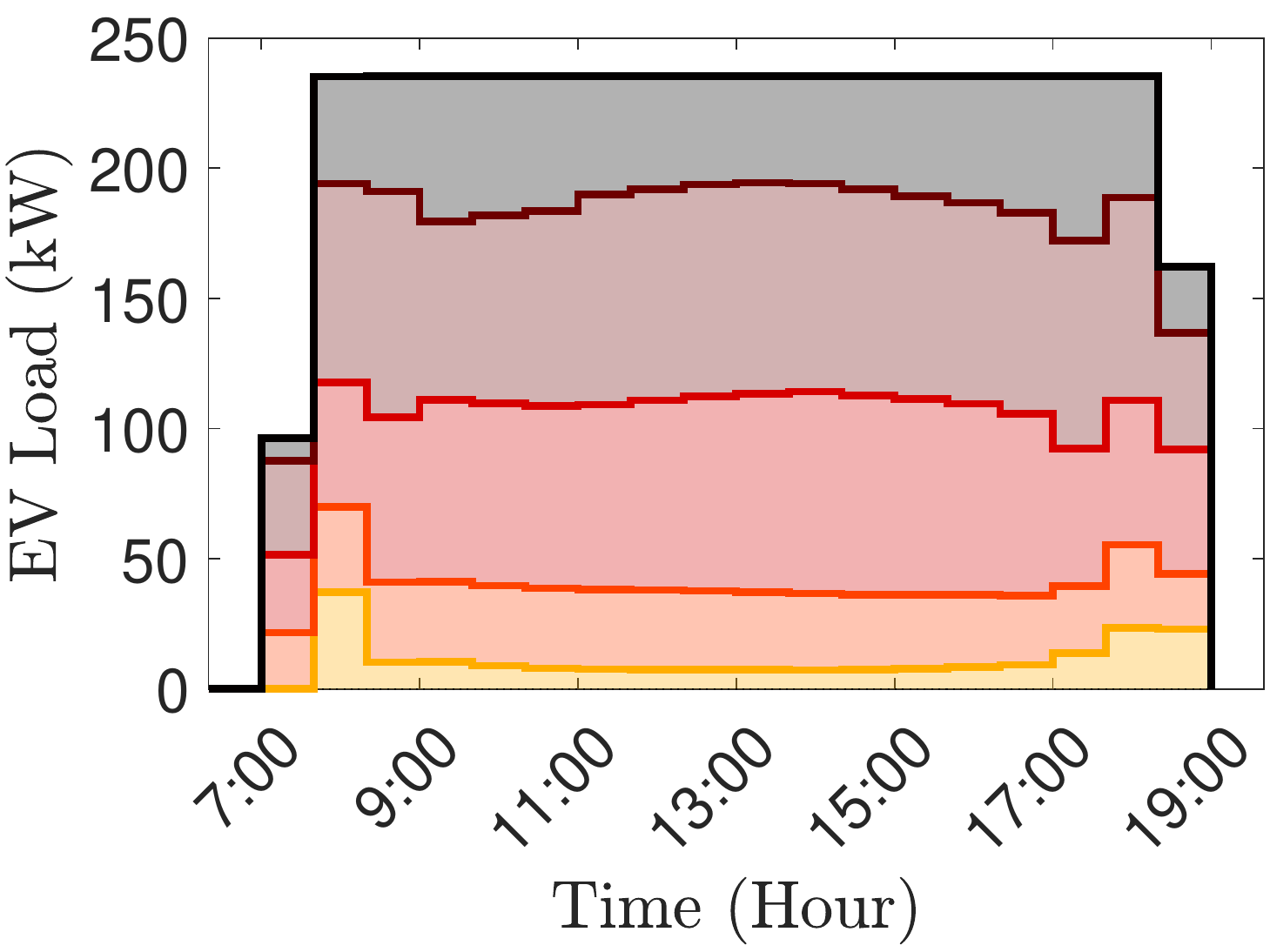}\label{fig:disagg_K}}\quad 
    \subfloat[Disaggregation of the EV load profile across the $N = 50$ individual EVs.]{\includegraphics[width=0.32\linewidth]{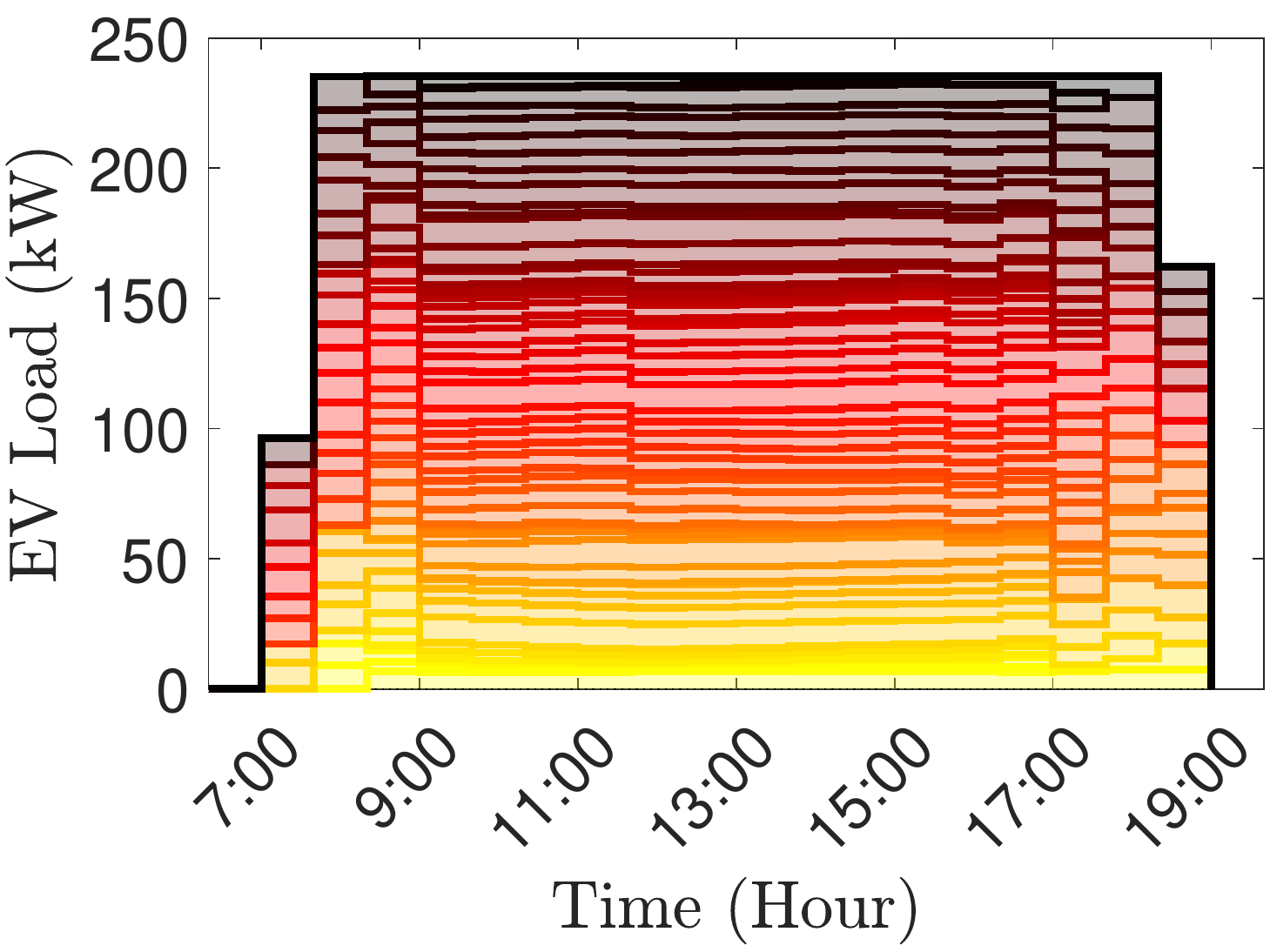}\label{fig:disagg_N}}
    \caption{Aggregate EV charging profile $u$ under a $K=5$ multi-battery approximation and its disaggregation across $K=5$ virtual batteries ($u=\sum_{k\in\Kcal}\mu_k+\sigma_k \tilde{u}_k$), and across $N=50$ EVs ($u=\sum_{i\in\Ncal}\gamma_i+\sum_{k\in\Kcal}\Gamma_{k,i}\tilde{u}_k$), where $\tilde{u}_k\in\Bset_k$ for $k=1,\dots,K$.
    The disaggregated load profiles are stacked and the solid black line (in (b) and (c)) corresponds to the total load.
    }
    \label{fig:disagg}
\end{figure*}

In this section, we evaluate the accuracy of the multi-battery approximation by varying the number of base sets used in the multi-battery approximation. We measure the performance of the resulting approximations in the context of a simple aggregate peak load minimization problem. 

More specifically, given an aggregate flexibility set $\Uset$, the aggregator seeks to solve the following optimization problem:
\begin{align*}
     \text{minimize}  \ \, \|u\|_\infty  \ \, \text{subject to} \ \ u\in\Uset, 
\end{align*}
whose optimal value we denote by $J^* := \min\{ \|u\|_\infty \,|\, u\in\Uset \}$.
The optimal solution to this peak shaving problem will be a feasible aggregate charging profile that meets the energy demands of all EVs while using minimal aggregate charging rates.
Such controlled charging is more desirable than unmanaged charging (where EVs charge at the maximum rate until their energy demand is met), which can result in a large aggregate peak load, as illustrated in Figure \ref{fig:all_loads}.

We conduct 100 independent experiments, where, for each experiment, we simulate the charging requirements of $N=50$ EVs and construct their corresponding individual flexibility sets.
Given these randomly sampled sets, we compute different multi-battery approximations by varying the number of base sets $K$ from one to five.
These base sets are constructed as described in Section \ref{sec:cluster}, where the single-battery set ($K=1$) corresponds to the base set used in \cite{zhao2017geometric} and \cite{altaha2022efficient}. 
Then, for each multi-battery approximation, we solve the peak shaving problem over the approximation of the aggregate flexibility set, and calculate the suboptimality gap incurred by its corresponding solution as 
\begin{align*}
    \text{gap}(K) \ = 100 \times \left(\frac{J(K)-J^*}{J^*}\right),
\end{align*}
where $J(K):=\min\{ \|u\|_\infty \,|\, u\in\Bset(K)\}$ and $\Bset(K)$ is the multi-battery flexibility set composed of $K$ base sets.
The multi-battery $\Bset(K)$ is computed by solving problem \eqref{eq:convex_program} using the Euclidean norm in the objective function, and using the sparse representations of the individual flexibility sets  defined in Remark \ref{rem:sparse} for faster computations.

In Figure~\ref{fig:optval}, we plot the suboptimality gap as a function of the number of base sets $K$.
It can be seen that there is a significant improvement in performance when going from a solution obtained with a single-battery approximation ($K=1$) to a solution obtained with a multi-battery approximation ($K>1$).

\begin{figure}
    \centering
    \includegraphics[width=.7\columnwidth]{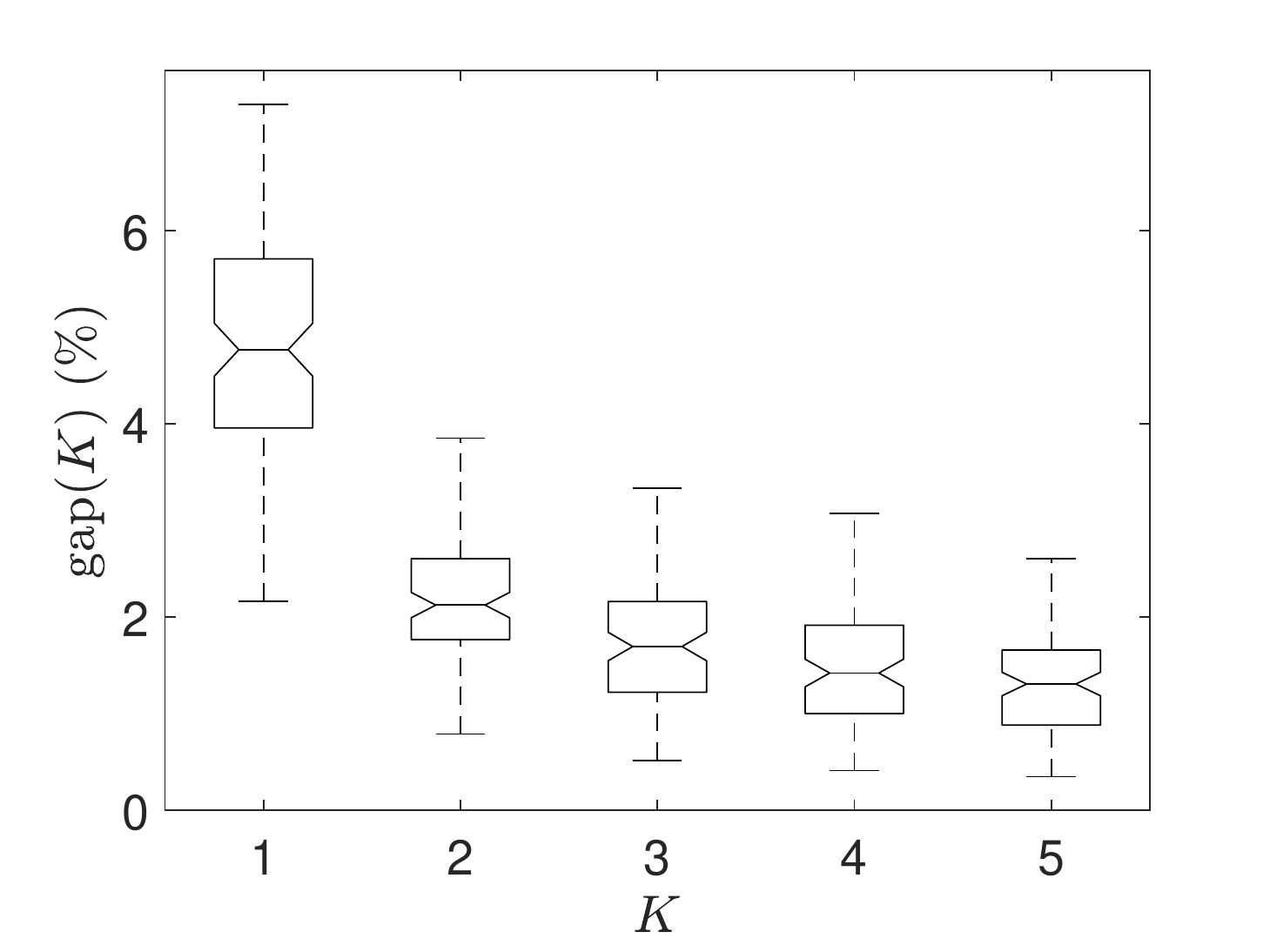}
    \caption{Suboptimality gap versus the number of base sets $K$. The whiskers determine the min-max range, the box delimits the interquartile range and the notch determines the median.
    }
    \label{fig:optval}
\end{figure}

For one randomly selected experiment, we depict in Figure~\ref{fig:disagg} the suboptimal load profile obtained from  optimizing over a multi-battery approximation with $K=5$ base sets, along with a disaggregation of this profile across the $K$ virtual batteries and across the $N$ EVs.
While the aggregate load is flattened (such that the aggregate charging rate is minimal), the disaggregated charging profiles are not necessarily flat.
As the connection windows of the EVs do not overlap completely, it can be advantageous to charge at a higher rate when there are fewer EVs connected.

\section{Conclusion} \label{sec:conclusion}

 In this paper, we presented a multi-battery modelling approach to  internally  approximating the Minkowski sum of a collection of heterogeneous \ EV flexibility sets.
This model yields a novel class of approximating polytopes composed of the sum of transformations of multiple given convex polytopes (termed \textit{base sets}), which are selected using a clustering algorithm to capture the different geometric shapes represented by the different  individual flexibility sets.
The number of base sets (a user-defined parameter) controls the multi-battery model's complexity as well as the quality of the resulting approximations.
We show how to conservatively approximate the problem of computing a multi-battery inner approximation which minimizes the Hausdorff distance to the aggregate flexibility set using a convex program whose size scales polynomially with the number and dimension of the individual flexibility sets and base sets.
Using a peak load minimization application, the proposed multi-battery approximation methods are shown to improve upon the performance of single-battery  approximations. 

As a direction for future research, it would be interesting to extend the modeling framework and  methods developed in this paper to accommodate more realistic EV charging characteristics, such as battery energy dissipation, charging inefficiencies, battery degradation, and other factors affecting the charging dynamics and constraints of an EV battery.

% *** BIB FORMATTING ***
\bibliographystyle{IEEEtran}
\bibliography{references}{\markboth{References}{References}}

\end{document}